\documentclass[journal]{IEEEtran}

\usepackage{bm}
\usepackage{float}
\usepackage[colorlinks,linkcolor=blue]{hyperref}
\usepackage[utf8]{inputenc}
\usepackage{amsmath}
\usepackage{graphicx}
\usepackage{amssymb}
\usepackage{amsthm}
\usepackage{verbatim}
\usepackage{epsfig}
\usepackage[labelfont=bf]{caption}
\usepackage{subcaption}
\usepackage{enumerate}
\usepackage{lmodern}
\usepackage{booktabs}
\usepackage{stfloats}
\usepackage{amssymb}
\usepackage{amsthm}
\usepackage{graphicx}
\usepackage{float}
\usepackage{amsmath}
\usepackage{amsmath,bm}
\usepackage{color}
\usepackage{multirow}
\usepackage{mathrsfs}
\usepackage{lineno,hyperref}
\usepackage{lipsum}
\usepackage{stfloats}

\usepackage{dsfont}
\usepackage{mathtools}

\newtheorem{definition}{Definition}

\newtheorem{remark}{Remark}

\newtheorem{corollary}{Corollary}
\newtheorem{proposition}{Proposition}
\newcommand{\mrm}{\mathrm}
\def\begquo{\begin{quote}}
	\def\endquo{\end{quote}}
\def\begequarr{\begin{eqnarray}}
	\def\endequarr{\end{eqnarray}}
\def\begequarrs{\begin{eqnarray*}}
	\def\endequarrs{\end{eqnarray*}}
\def\begarr{\begin{array}}
	\def\endarr{\end{array}}
\def\begequ{\begin{equation}}
	\def\endequ{\end{equation}}
\def\lab{\label}
\def\begdes{\begin{description}}
	\def\enddes{\end{description}}
\def\begenu{\begin{enumerate}}
	\def\begite{\begin{itemize}}
		\def\endite{\end{itemize}}
	\def\endenu{\end{enumerate}}

\def\lef[{\left[\begin{array}}
	\def\rig]{\end{array}\right]}

\def\begcen{\begin{center}}
	\def\endcen{\end{center}}
\def\begrem{\begin{remark}\rm}
	\def\endrem{\end{remark}}
\def\begdef{\begin{definition}}
	\def\enddef{\end{definition}}
\def\begpro{\begin{propositionosition}}
	\def\endpro{\end{propositionosition}}
\def\begfac{\begin{fact}}
	\def\endfac{\end{fact}}
\def\begass{\begin{assumptionption}}
	\def\endass{\end{assumptionption}}

\def\begmat#1{\begin{bmatrix}#1\end{bmatrix}}
\def\begali#1{\begin{align}{#1}\end{align}}

\def\L2e{{\cal L}_{2e}}

\def\adj{\mbox{adj}}

\def\max{{\mbox{max}}}

\usepackage{color}

\usepackage[prependcaption,colorinlistoftodos]{todonotes}

\title{PMU-based dynamic state and parameter estimation for dynamic
	security assessment in power systems - Ultimate boundedness in the presence of measurement noise}
\author{Nicolai Lorenz-Meyer, René Suchantke, and Johannes Schiffer
\thanks{N. Lorenz-Meyer is with Brandenburg University of Technology Cottbus-Senftenberg, 03046 Cottbus, Germany (e-mail: lorenz-meyer@b-tu.de).}
\thanks{R. Suchantke is with 50Hertz Transmission GmbH, 10557 Berlin, Germany (e-mail: rene.suchantke@50hertz.com).}
\thanks{J. Schiffer is with Brandenburg University of Technology Cottbus-Senftenberg, 03046 Cottbus, Germany and Fraunhofer Research Institution for Energy Infrastructures and Geothermal Systems (IEG), 03046 Cottbus, Germany (e-mail: schiffer@b-tu.de).}
}
\begin{document}
	
	\maketitle
	
\begin{abstract}
Dynamic state and parameter estimation methods for dynamic security assessment in power systems are becoming increasingly important for system operators.
		Usually, the data used for this type of applications stems from phasor measurement units (PMUs) and is corrupted by noise. In general, the impact of the latter may significantly deteriorate the estimation performance. 
			This motivates the present work, in which 
			it is proven that the state and parameter estimation method proposed by part of the authors in \cite{lorenz-meyer_pmu-based_2020-1} and extended in \cite{lorenz-meyer_dynamic_2022} features the property that the estimation errors are ultimately bounded
			in the presence of PMU measurement data corrupted by bounded noise. The analysis is conducted for the third-order flux-decay model of a synchronous generator and holds independently of the employed automatic voltage regulator and power system stabilizer (if present). The analysis is illustrated by simulations.
		\end{abstract}

	\section{Introduction}
\subsection{Motivation}	
Due to the massive introduction of renewable, power-electronics interfaced components, as well as the implementation of novel demand-response technologies, power systems undergo major changes~\cite{zhao_power_2019, winter_pushing_2015}. Also, an increased loading of grid assets yields system operation closer to the stability limit~\cite{ulbig_impact_2014}. Consequently, the reliable monitoring of the system states as well as the accurate estimation of dynamic system parameters is crucial to ensure a stable and reliable system operation~\cite{zhao_et_al_power_2021}.

Additionally, analyzing various stability phenomena is an important part of a transmission system operator's (TSOs) responsibilities. During the last years, stability analyses increased in importance and complexity, with new stability phenomena arising, specifically converter-driven stability and resonance stability~\cite{hatziargyriou_et_al_definition_2021}. To be able to perform accurate dynamic studies of a power system, the network elements need to be represented precisely in the simulation tool. However, many network elements, which play a crucial role in stability analyses, are not owned by a TSO, such as generation units (e.g. power plants and wind farms) and loads, among others. Therefore, the availability of dynamic data for the TSO is oftentimes very limited. Documentation of dynamic behavior may not be accessible, unusable due to bad or patchy documentation or is not available at all.
In this case, dynamic state and parameter estimations can help to recover crucial dynamic data of an asset, thereby supporting and improving system model validity. Moreover, TSOs world-wide expand their phasor measurement unit (PMU) infra\-structure. PMUs perform a GPS-synchronized measurement of current and voltage phasors and usually record those values at rates up to~50 or~60 times a second (once per cycle) \cite{zhou_dynamic_2015}. 
The main motivation from a TSOs perspective is to utilize the abundant measurements provided by PMUs in combination with the estimation algorithm studied in the present paper, for generation units, whose dynamic data is unknown. 
This is further emphasized as accurate knowledge of dynamic behavior of large conventional power plants will become even more important, when the number of large units decreases in the power grid of the future.

\subsection{Contributions}
Building on the mixed algebraic and dynamic state observation method for multi-machine power systems presented in \cite{lorenz-meyer_pmu-based_2020-1} and extended in \cite{lorenz-meyer_dynamic_2022}, we provide an analytical analysis of the impact of measurement noise on the method's performance in terms of the measurement error. 
In this setting, our main contributions are three-fold:
\begin{enumerate}
	\item We prove that the estimation errors of the algebraic estimator for the rotor angle and the quadrature-axis internal voltage remain ultimately bounded in the presence of PMU measurement data corrupted by bounded noise. Furthermore, we show that - even though the estimation algorithm is nonlinear - the algebraic estimate can be expressed as a nominal part with an affine disturbance originating from the influence of the measurement noise.
	\item By utilizing the bounds on the algebraic estimate obtained in 1., we prove that the trajectories of the estimation error of the adaptive observer for the relative shaft speed as well as the damping and inertia constants remain ultimately bounded. For this, we derive a linear time-varying (LTV) system description of the estimation error dynamics. We then utilize the structure of the LTV system together with the exponential stability of the noiseless adaptive observer shown in \cite{lorenz-meyer_pmu-based_2020-1} to prove the ultimate boundedness of the trajectories. 
	We illustrate the impact of PMU measurements with Gaussian and non-Gaussian noise on the method's performance in simulations using the New England IEEE-39 bus system \cite{hiskens_ieee_nodate}.
\end{enumerate} 
The remainder of the paper is structured as follows. The mathematical model of the multi-machine power system is introduced in Section 2. The influence of measurement noise on the dynamic state and parameter estimation method is  derived in Section 3 and illustrated in Section 4 using simulation results. Concluding remarks are given in Section 5.

\section{Model of the electric power system}
In this note, an electric power network with $N~>~1$ synchronous generators (SGs) is considered. Each of the SGs is described by the well-known third-order flux-decay model \cite[Eq. (7.176-7.178)]{sauer_power_2006}, see also \cite[Eq. (11.108)]{machowski_power_2008}, \cite[Eq. (3.3), (3.5), (3.15)]{van_cutsem_voltage_1998}, which is given by
\begin{subequations}
	\lab{model}
	\begali{
		\label{eq:x1}	\dot{x}_{1}&= x_{2}  ,\\
		\label{eq:x2}	\dot{x}_{2}&=\frac{\omega_\mrm{s}}{2H}(T_\mrm{m}-T_\mrm{e}-Dx_{2})  ,\\
		\dot{x}_{3} &= \frac{1}{T_\mrm{d0}'}\left(-x_{3}-(x_\mrm{d}-x_\mrm{d}')I_\mrm{t}\sin(x_1- \phi_\mrm{t})+E_\mrm{f}\right).}
\end{subequations}
with $T_{\mrm{m}}$ denoting the mechanical power, $T_{\mrm{e}}$ denoting the electrical air-gap torque, $E_{\mrm{f}}$ denoting the field voltage, $\phi_\mrm{t}$ denoting the terminal current angle, $I_{\mrm{t}}$ and $V_{\mrm{t}}$ denoting the terminal current and voltage magnitude, respectively. Furthermore, the known (positive) model parameters are the direct-axis reactance $x_{\mrm{d}}$, the direct-axis transient reactance $x_{\mrm{d}}'$, while the damping factor $D$, the inertia constant $H$ and  direct-axis transient open-circuit time constant $T_{\mrm{d0}}'$ are assumed unknown.
In \eqref{model}, the unknown state vector is defined as
\begin{align}
	x& := \begin{bmatrix} x_{1} & x_{2} & x_{3} \end{bmatrix}^\top = \begin{bmatrix} \delta & \omega-\omega_\mrm{s} & E_\mrm{q}' \end{bmatrix}^\top,
\end{align}
with $\delta$ denoting the rotor angle, $\omega$ denoting the shaft speed, $\omega_\mrm{s}$ denoting the nominal synchronous speed and $E_{\mrm{q}}'$ denoting the quadrature-axis internal voltage.
By defining the unknown constants
\begin{equation}
	\begin{split}
		a_{1}&=\frac{\omega_s D}{2H}, \ \ a_{2}= \frac{\omega_s}{2H},
	\end{split} \label{Eq:Paramaters_a1a2}
\end{equation}	
the model \eqref{model} of an individual SG can be compactly written as 
\begin{subequations}
	\lab{x}
	\begali{
		\lab{x1}
		\dot x_{1}&=x_{2},\\
		\lab{x2}	
		\dot x_{2}&= -a_{1} x_{2}+ a_{2}(T_\mrm{m} - T_\mrm{e}) ,\\
		\label{x3}
		\dot x_{3} &= \frac{1}{T_\mrm{d0}'}\left(-x_{3}-(x_\mrm{d}-x_\mrm{d}')\sin(x_1- \phi_\mrm{t})+E_\mrm{f}\right).
	}
\end{subequations}
\begin{figure}
	\centering
	\includegraphics[width=1\linewidth]{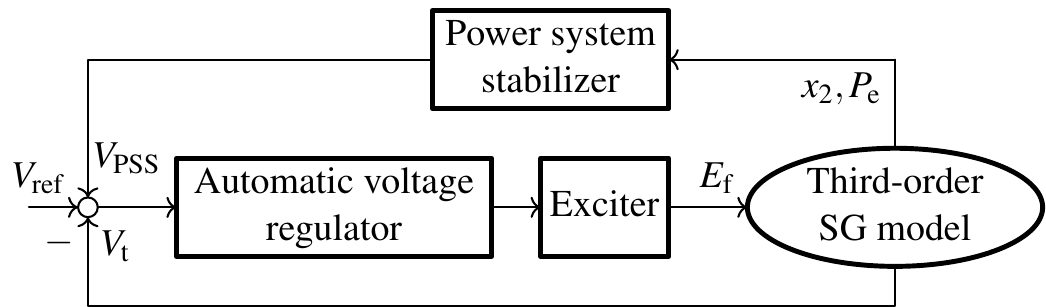}
	\caption{Overall SG model including the excitation system as well as the PSS. 
		\label{fig:SG_model}
	}
\end{figure}
Moreover, the stator algebraic equation is expressed in accordance with  \cite[Eq. (5.134)]{sauer_power_2006} as
\begin{equation}
	\begin{split}
		\label{eq:stator_alg}
		\mrm{j} x_{3}\mrm{e}^{\mrm{j}(x_{1}-\frac{\pi}{2})} &= (R_\mrm{s} + \mrm{j} x_\mrm{d}') I_\mrm{t} \mrm{e}^{\mrm{j}\phi_\mrm{t}} + V_\mrm{t} \mrm{e}^{\mrm{j}\theta_\mrm{t}} -\\& \ \ \ \  - (x_\mrm{q} - x_\mrm{d}') \cos(x_1- \phi_\mrm{t})\mrm{e}^{\mrm{j}(x_{1}-\frac{\pi}{2})}, 
	\end{split}
\end{equation}
where $x_\mrm{q}$ denotes the quadrature-axis reactance, $R_\mrm{s}$ the stator resistance and $\theta_\mrm{t}$ the terminal voltage angle. 
\begin{remark}
	The SG is usually equipped with an excitation system to control the terminal voltage $V_\mrm{t}$. This typically consists of an automatic voltage regulator (AVR) and an exciter to provide the required field voltage $E_\mrm{f}$. Additionally, a power system stabilizer (PSS) can be added to the excitation system to help dampen power swings in the system \cite[Chapter 2.3]{machowski_power_2008}. In this work, knowledge of the time-varying field voltage is not required to estimate the states nor the unknown parameters of the third-order flux decay model. Consequently and for the sake of clarity, the potentially high dimensional models of the excitation system and the PSS are not included in \eqref{model}. Though, in Figure \ref{fig:SG_model}, we illustrate the cascaded structure of the overall SG model including the excitation system and the PSS. 
\end{remark}	 
\section{Dynamic state and parameter estimation in the presence of  measurement noise}
In the following, we analyze the effect of measurement noise on the dynamic state and parameter estimation method presented in \cite{lorenz-meyer_pmu-based_2020-1} and extended in \cite{lorenz-meyer_dynamic_2022}. For this, we assume that the PMU measurements available at the generator terminal are affected by bounded noise, i.e.,
\begin{equation}
	\begin{split}
		\lab{y}
		y_{1}&=V_{\mrm{t}} + w_1, \quad
		y_{2}=\theta_\mrm{t} + w_2,\\
		y_{3}&=I_{\mrm{t}} + w_3,  \quad
		y_{4}=\phi_\mrm{t} + w_4 \\
		y_{5}&=P_{\mrm{t}} + w_5,  \quad
		y_{6}=Q_\mrm{t} + w_6.
	\end{split}
\end{equation}
Here $|w_i|\leq w_i^\mrm{max}, i = \{1,2, \dots ,6\}$, denotes the bounded measurement noises and $w_i^\mrm{max}\geq 0$ their upper bounds. $P_{\mrm{t}}$ and $Q_{\mrm{t}}$ denote the terminal active and reactive power, respectively. 
Furthermore, we assume that the value of the mechanical torque $T_{\mrm{m}}$ is available up to an affine error, i.e.,
\begin{equation}
	\label{eq:Tm_est}
	\hat{T}_\mrm{m}=T_{\mrm{m}} + w_\mrm{Tm}, \quad |w_\mrm{Tm}| \leq w_\mrm{Tm}^\mrm{max}.
\end{equation}	
\begin{remark}
	Depending on the type of power plant for which the method is used, the mechanical torque can be modeled, e.g. by a power control governor and turbine model (see \cite{lorenz-meyer_dynamic_2022}). Consequently, the mechanical torque can be calculated online in dependence on available PMU measurements \eqref{y}. 
\end{remark}
For the readers' convenience, we briefly recall the key results of the previously presented dynamic state and parameter estimation method whenever necessary. We denote these nominal equations for the noiseless case with superscript $\{\cdot\}^\mrm{nom}$. Furthermore, in Figure \ref{fig:Obs_Overview} we summarize the overall structure of the previously proposed method. 
\begin{figure}
	\centering
	\includegraphics[width=1\linewidth]{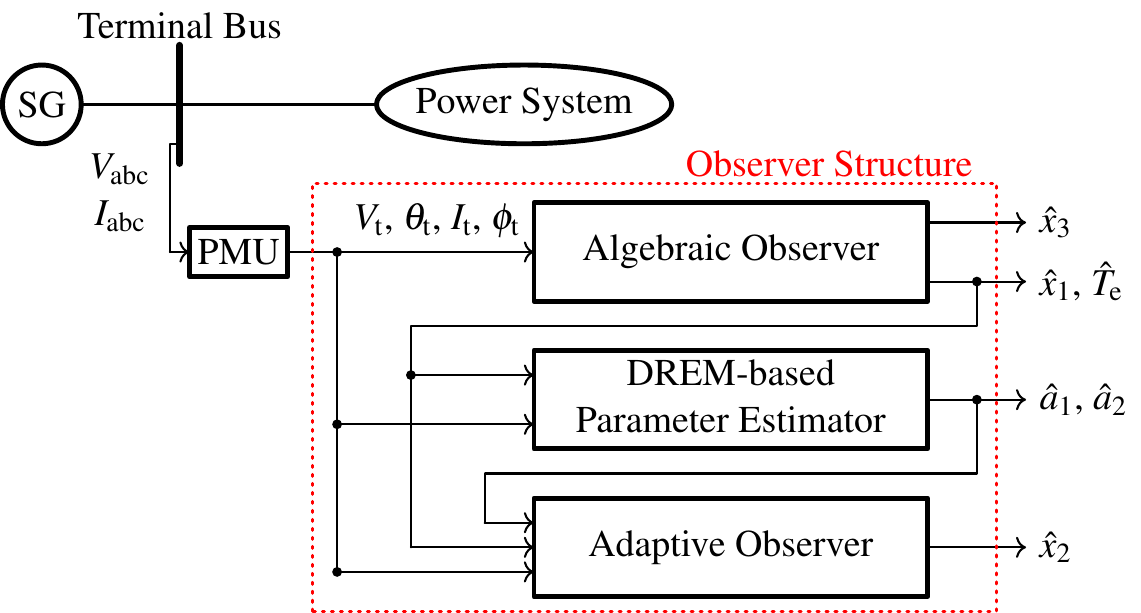}
	\caption{Structure of the previously presented dynamic state and parameter estimation method. 
		\label{fig:Obs_Overview}
	}
\end{figure}
\subsection{Influence of measurement noise on the algebraic observer for $x_1$ and $x_3$}
We study the effect of bounded measurement noise on the algebraic observer presented in \cite{lorenz-meyer_dynamic_2022}.
For this and as shown in \cite{lorenz-meyer_dynamic_2022} with the purpose of algebraically estimating the states $x_1$ and $x_3$ of the system \eqref{x}, the stator algebraic equation \eqref{eq:stator_alg} can be rewritten as 
\begin{equation}
	\label{eq:staor_alg__mani_noiseless}
	\begin{split}
		& \left((x_\mrm{q} - x_\mrm{d}') \cos\left(\frac{\pi}{2}-x_{1}+\phi_\mrm{t}\right)I_\mrm{t} +  x_{3}\right) \mrm{e}^{\mrm{j}x_{1}}  = \\& \ \ \ \ \ \ = (R_\mrm{s} + \mrm{j} x_\mrm{q}) I_\mrm{t} \mrm{e}^{\mrm{j}\phi_\mrm{t}} + V_\mrm{t} \mrm{e}^{\mrm{j}\theta_\mrm{t}} \eqqcolon \vec{\psi}^\mrm{nom} .
	\end{split}	
\end{equation}
With noiseless PMU measurements the nominal right-hand side of \eqref{eq:staor_alg__mani_noiseless}, i.e., $\vec{\psi}^\mrm{nom}$, is known and thus the states $x_1$ and $x_3$ of system \eqref{x} are uniquely determined by (see \cite[ Lemma 1]{lorenz-meyer_dynamic_2022})
\begin{equation}
	\begin{split}
		\label{eq:arxiv_alg_obs}
		x_{1} &= \arg\{\vec{\psi}^\mrm{nom}\}   ,\\
		x_{3} &= |\vec{\psi}^\mrm{nom}| - (x_\mrm{q}-x'_\mrm{d})\cos\left(\frac{\pi}{2}-x_{1}+\phi_\mrm{t}\right)I_\mrm{t}  .
	\end{split}	
\end{equation}
We replace the true values of $V_\mrm{t}, \theta_\mrm{t}, I_\mrm{t}$ and $\phi_\mrm{t}$ by the noisy measurements \eqref{y} and thus rewrite \eqref{eq:staor_alg__mani_noiseless} as 
\begin{equation}
	\begin{split}
		\label{eq:stator_alg_manip}
		& \left((x_\mrm{q} - x_\mrm{d}') \cos\left(\frac{\pi}{2}-\hat{x}_{1}+y_4\right)y_3  +  \hat{x}_{3}\right) \mrm{e}^{\mrm{j}\hat{x}_{1}} = \\& \ \ \ \ \ \ =  (R_\mrm{s} + \mrm{j} x_\mrm{q}) y_3\mrm{e}^{\mrm{j}y_4} + y_1 \mrm{e}^{\mrm{j}y_2} \eqqcolon \vec{\psi}  .
	\end{split}	
\end{equation}
Hence, we obtain an estimate of the states $x_{1}$ and $x_{3}$ as
\begin{equation}
	\begin{split}
		\label{eq:alg_obs_noise}
		\hat{x}_{1} &= \arg\{\vec{\psi}\}   ,\\
		\hat{x}_{3} &= |\vec{\psi}| - (x_\mrm{q}-x'_\mrm{d})\cos\left(\frac{\pi}{2}-\hat{x}_{1}+y_4\right)y_3.
	\end{split}
\end{equation}
We have the following result, characterizing the estimation errors
\begin{equation}
	\label{eq:err:alg_obs}
	\begin{split}
		\tilde{x}_1 = \hat{x}_1 - x_1 , \quad  \tilde{x}_3 = \hat{x}_3 - x_3 ,
	\end{split}
\end{equation} in the presence of bounded noise.
\begin{proposition}
	\lab{pro:bound_alg_obs}
	Consider the states $x_{1}$ and $x_{3}$ of the system \eqref{x}, and the algebraic estimation formula given in \eqref{eq:alg_obs_noise}. Under the assumption of noisy measurements as in \eqref{y}, there exist positive constants $w^{max}_{x1}$ and $w^{max}_{x3}$ such that
	$$|\tilde{x}_1|\leq  w^{max}_{x1},\quad |\tilde{x}_3|\leq  w^{max}_{x3}.$$
	That is, the estimation errors are ultimately bounded.
\end{proposition}
As Proposition \ref{pro:bound_alg_obs} shows, bounded measurement noise results in a bounded estimation error for the algebraic observer \eqref{eq:alg_obs_noise}. This important robustness property cannot be guaranteed in general for nonlinear observers \cite{shim_nonlinear_2016}, which underpins the relevance of asserting it in Proposition \ref{pro:bound_alg_obs}.

\begin{proof}
	To analyze the properties of the estimation errors $\tilde{x}_1$ and $\tilde{x}_3$, we utilize \eqref{y} and expand $\vec{\psi}$ defined in \eqref{eq:stator_alg_manip} as
	\begin{align*}
		\vec{\psi}=  (R_\mrm{s} + \mrm{j} x_\mrm{q}) (I_\mrm{t} + w_3)\mrm{e}^{\mrm{j}(\phi_\mrm{t}+w_4)} + (V_\mrm{t}+w_1) \mrm{e}^{\mrm{j}(\theta_\mrm{t}+w_2)}.
	\end{align*}
	
	Let $a$ be either $\theta_\mrm{t}$ or $\phi_\mrm{t}$. Then, the terms $\mrm{e}^{\mrm{j}(\theta_\mrm{t}+w_2)}$ and $\mrm{e}^{\mrm{j}(\phi_\mrm{t}+w_4)}$ can be expanded as follows:
	\begin{equation*}
		\begin{split}
			\mrm{e}^{\mrm{j}(a+w_i) }&=\mrm{e}^{\mrm{j}a}\mrm{e}^{\mrm{j}w_i} \\
			&= \cos(a)  + \mrm{j}\sin(a)  +\\ &\ \ \ \ +  w_i(- \sin(a)  + \mrm{j}\cos(a) ) + \mathcal{O}\left(w_i^2\right)
			\\& = \mrm{e}^{\mrm{j}a}+\mrm{j} w_i\mrm{e}^{\mrm{j}a}+ \mathcal{O}\left(w_i^2\right).
		\end{split}	
	\end{equation*}
	Here we have used the small angle approximation for the noise in the phase angles $w_i$ as 
	$$ \cos (w_i) = 1 +\mathcal{O}(w_i^2)\ \ \ , \ \ \ \sin (w_i) = w_i +\mathcal{O}(w_i^3).$$ 
	Hence, we can write $\vec{\psi}$ as
	\begin{equation}
		\label{eq:psi}
		\begin{split}
			\vec{\psi} & = (R_\mrm{s} + \mrm{j} x_\mrm{q}) (I_\mrm{t} + w_3)\left(\mrm{e}^{\mrm{j}\phi_\mrm{t}} + \mrm{j} w_4\mrm{e}^{\mrm{j}\phi_\mrm{t}} \right) + \\& \ \ \ \ + (V_\mrm{t}+w_1) \left(\mrm{e}^{\mrm{j}\theta_\mrm{t}} + \mrm{j} w_2\mrm{e}^{\mrm{j}\theta_\mrm{t}} \right) + \\ & \ \ \ \ + \mathcal{O}\left((I_\mrm{t}+w_3)w_2^2+(V_\mrm{t}+w_1)w_4^2\right)\\ 
			& = \vec{\psi}^\mrm{nom} + \left(w_3 + \mrm{j} (I_\mrm{t}w_4+ w_3w_4) \right) (R_\mrm{s} + \mrm{j} x_\mrm{q}) \mrm{e}^{\mrm{j}\phi_\mrm{t}}+ \\& \ \ \ \ + ((w_1 + \mrm{j}(V_\mrm{t} w_2 +w_1w_2) \mrm{e}^{\mrm{j}\theta_\mrm{t}} +\\ & \ \ \ \ + \mathcal{O}\left((I_\mrm{t}+w_3)w_2^2+(V_\mrm{t}+w_1)w_4^2\right).
		\end{split}	
	\end{equation}
	\begin{figure}
		\centering
		\includegraphics[width=1\linewidth]{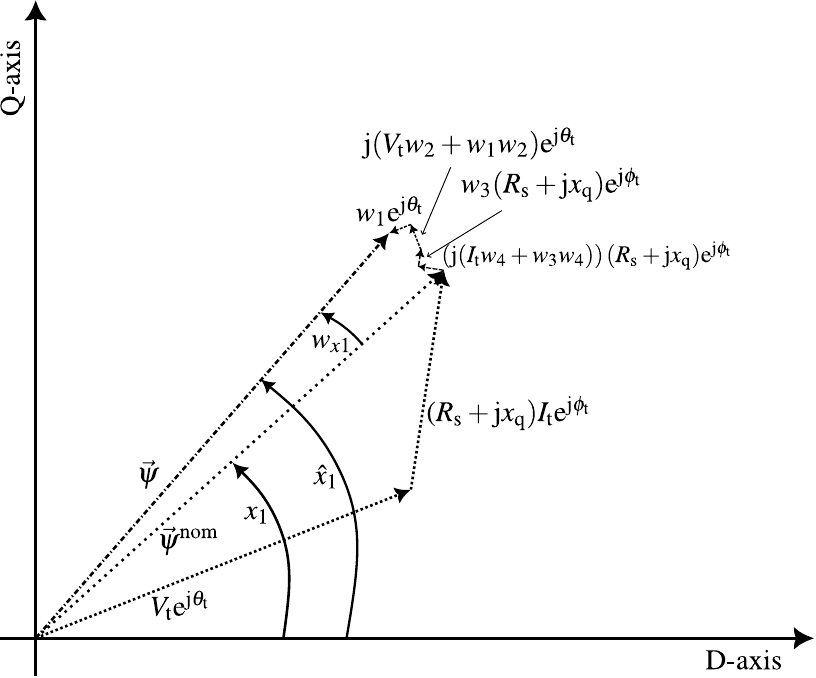}
		\caption{Phasor diagram of $\vec{\psi}$ and $\vec{\psi}^\mrm{nom}$ representing the size of the error in the estimate of $x_1$ caused by the noisy measurements \eqref{y}. In this representation the error of order $\mathcal{O}\left((I_\mrm{t}+w_3)w_2^2+(V_\mrm{t}+w_1)w_4^2\right)$ is not shown. 
			\label{fig:arg_psi}
		}
	\end{figure}
	The measurement noise causes a parallel and a tangential disturbance in the direction of the nominal phasors, i.e., in the direction of $\mrm{e}^{\mrm{j}\theta_\mrm{t}}$ and $(R_\mrm{s} + \mrm{j} x_\mrm{q})\mrm{e}^{\mrm{j}\phi_\mrm{t}}$. As can be seen in Figures \ref{fig:arg_psi} and \ref{fig:abs_psi}, the four disturbance terms in \eqref{eq:psi} add up and cause an affine error in the argument and the absolute value of $\vec{\psi}$. Thus, these can be expressed as
	\begin{subequations}
		\begin{align}
			\arg\{\psi\} &= \arg\{\vec{\psi}^\mrm{nom}\} + w_{x1}, \\
			|\psi| &= |\vec{\psi}^\mrm{nom}| + w_{\psi}, 	\label{eq:err_abs_psi}
		\end{align}
	\end{subequations}
	with the upper bounds $|w_{x1}|\leq w_{x1}^\mrm{max}$ and $|w_{\psi}|\leq w_{\psi}^\mrm{max}$. The existence of these bounds is guaranteed as the measurement noise as well as the terminal voltage and current are bounded by $|w_i|\leq w_i^\mrm{max}, i = \{1,2,\dots,6\}$, $|V_\mrm{t}|\leq V_{\mrm{t}}^\mrm{max}$ and $|I_\mrm{t}|\leq I_{\mrm{t}}^\mrm{max}$. Furthermore, the error term $\mathcal{O}\left((I_\mrm{t}+w_3)w_2^2+(V_\mrm{t}+w_1)w_4^2\right)$ is bounded following the same argument.  
	
	Hence, the estimate of the state $x_1$ is bounded as
	\begin{equation*} \label{eq:hatx1}
		\hat{x}_{1} = x_1 + w_{x1}.
	\end{equation*}
	Moreover, using the definition of the noisy measurements \eqref{y} as well as the derived expression for $|\psi|$ in \eqref{eq:err_abs_psi}, the estimate of the state $x_3$ given in \eqref{eq:alg_obs_noise} can be expressed as 
	\begin{equation*}
		\begin{split}
			\hat{x}_{3} &= - (x_\mrm{q}-x'_\mrm{d})\cos\left(\frac{\pi}{2}-{x}_{1} - w_{x1}+\phi_\mrm{t}+ w_4\right)(I_\mrm{t}+w_3) + \\ & \ \ \ \  +|\vec{\psi}^\mrm{nom}| + w_{\psi} .
		\end{split}
	\end{equation*}
	Some algebraic manipulations utilizing the trigonometric identity for the sum of angles 
	as well as the small angle approximation 
	\begin{equation*}
		\begin{split}
			\cos (w_4-w_{x1}) &= 1 +\mathcal{O}\left((w_4-w_{x1})^2\right) , \\
			\sin (w_4-w_{x1}) &= w_4-w_{x1} +\mathcal{O}\left((w_4-w_{x1})^3\right) ,
		\end{split}
	\end{equation*}
	yields
	\begin{equation*}
		\begin{split}
			\hat{x}_{3} &= x_3 + w_{\psi} + (x_\mrm{q}-x'_\mrm{d})\bigg( -\cos\left(\frac{\pi}{2}-{x}_{1} +\phi_\mrm{t}\right)w_3 + \\ & \ \ \ \ + \sin\left(\frac{\pi}{2}-{x}_{1} +\phi_\mrm{t}\right)(I_\mrm{t} + w_3)(w_{x1}-w_4) \bigg) + \\ & \ \ \ \ +\mathcal{O}\left((I_\mrm{t}+w_3)(w_4-w_{x1})^2\right) \\
			&=  x_3 + w_{x3}.
		\end{split}
	\end{equation*}
	The proof is completed by noting that $w_{x3}$ can be bounded using the triangular inequality as follows
	\begin{equation*}
		\begin{split}
			|w_{x3}| &\leq x_\mrm{q}-x'_\mrm{d}|\left( |w_3| + (|I_\mrm{t}| + |w_3|)(|w_{x1}|-|w_4|) \right) +\\ & \ \ \ \  +|w_{\psi}| + |\mathcal{O}\left((I_\mrm{t}+w_3)(w_4-w_{x1})^2\right) 
			\\ & \leq | x_\mrm{q}-x'_\mrm{d}|\bigg( w_3^\mrm{max} + (I_\mrm{t}^\mrm{max} w_3^\mrm{max})(w_{x1}^\mrm{max}-w_4^\mrm{max}) \bigg)+ \\& \ \ \ \ +w_{\psi}^\mrm{max} +\mathcal{O}\left((I_\mrm{t}+w_3)(w_4-w_{x1})^2\right)
			\\&\leq w_{x3}^\mrm{max}.
		\end{split}
	\end{equation*}
	\begin{figure}
		\centering
		\includegraphics[width=\linewidth]{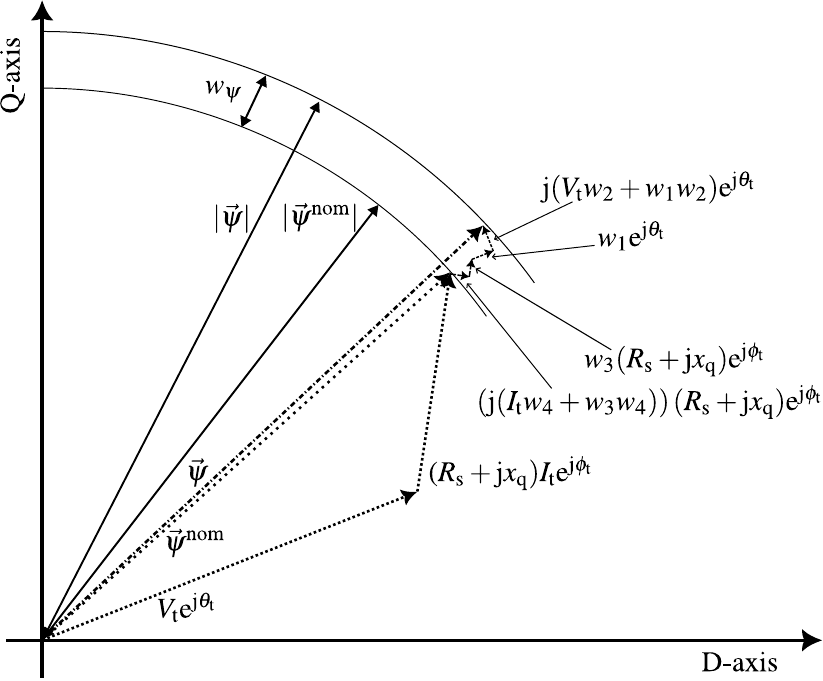}
		\caption{Phasor diagram of $\vec{\psi}$ and $\vec{\psi}^\mrm{nom}$ representing the size of the error in the estimate of the absolute value of $\vec{\psi}$ caused by the noisy measurements \eqref{y}. In this representation the error of order $\mathcal{O}\left((I_\mrm{t}+w_3)w_2^2+(V_\mrm{t}+w_1)w_4^2\right)$ is not shown.
			\label{fig:abs_psi}		}
	\end{figure}
\end{proof}
\subsection{Influence of measurement noise on the adaptive state and parameter estimator}
In the following, we study the influence of measurement noise on the adaptive state and parameter estimator presented in \cite{lorenz-meyer_pmu-based_2020-1, lorenz-meyer_dynamic_2022}. For this, we recall the adaptive observer for the state $x_2$ as well as for the unknown parameters $ \bm{\theta}= \begin{bmatrix}a_1 &a_2\end{bmatrix} ^\top $  :
\begin{equation*}
	\begin{split} 
		\dot{x}^{\mrm{I}}_{2} &= - ({\hat \theta_{1}}+k) ({x}^{I}_{2} + k x_{1}) +\hat \theta_{2} (T_\mrm{m}-T_\mrm{e}) ,\\
		\hat{x}_{2} &= {x}^\mrm{I}_{2} + k {x}_{1},\\
		\dot {\hat \theta}_{i}&=- \gamma_{i} \Delta^\mrm{nom}(\Delta^\mrm{nom}  \hat \theta_{i} -\mathcal{Z}_{i}^\mrm{nom}),\;i= \{1,2\}, 
	\end{split} 
\end{equation*}
where $\gamma_i>0$ and $k>0$ denote constant tuning parameters. $\Delta^\mrm{nom}$ and $\mathcal{Z}_{i}^\mrm{nom}$ are the nominal part of the regressor and the output after applying the 
dynamic regressor and extension (DREM) principle, respectively (see \cite{lorenz-meyer_pmu-based_2020-1, lorenz-meyer_dynamic_2022}).

To analyze the behavior of this adaptive observer in the case of PMU measurements affected by bounded noise, we replace the nominal values of $x_1, \Delta^\mrm{nom}, \mathcal{Z}_{i}^\mrm{nom}, T_\mrm{m} $ and  $T_\mrm{e}$ by their estimates 
\begin{subequations}
	\label{eq:ad_obs}
	\begali{
		\notag	\dot{x}^{\mrm{I}}_{2} &= - ({\hat \theta_{1}}+k) ({x}^{I}_{2} + k \hat{x}_{1}) +\hat \theta_{2} (\hat{T}_\mrm{m}-\hat{T}_\mrm{e}) ,\\
		\label{eq:obsx2}	\hat{x}_{2} &= {x}^\mrm{I}_{2} + k \hat{x}_{1},\\
		\label{eq:para_est}		\dot {\hat \theta}_{i}&=- \gamma_{i} \Delta(\Delta  \hat \theta_{i} -\mathcal{Z}_{i}),\;i=1,2.}
\end{subequations}
We have the following result, characterizing the estimation errors $\tilde{x}_2 = \hat{x}_2 - x_2$ and $\tilde{\theta}_{i} = \hat{\theta}_i - \theta_i, i = \{1,2\}$ in the presence of bounded noise.
\begin{proposition}
	\lab{pro:bound_adap_obs}
	Consider the state $x_{2}$ and the parameters ${\theta}_i, i = \{1,2\}$ of the system \eqref{x} as well as the adaptive observer given in \eqref{eq:ad_obs}. Under noisy measurements as in \eqref{y}, considering the estimate of the mechanical torque given in \eqref{eq:Tm_est} and assuming sufficiently small noise bounds, if the regressor $\Delta$ is of persistence of excitation, then there exist positive constants $w^{\mrm{max}}_{x2}$ and $w^{\mrm{max}}_{\theta},$ such that 
	\begin{subequations}
		\begin{align}
			|\tilde{x}_2| &\leq w_{x2}^\mrm{max} , \\
			|\tilde{\theta}_i| &\leq w_{\theta}^\mrm{max}, i = \{1,2\}.
		\end{align}
	\end{subequations}
	That is, the estimation errors are ultimately bounded.
\end{proposition}
\begin{proof}
	By utilizing \eqref{eq:alg_obs_noise} to estimate $x_1$ and $x_3$ and using the noisy measurements \eqref{y}, the electrical air-gap torque $T_\mrm{e}$ can be calculated via (see \cite{lorenz-meyer_dynamic_2022} for the noiseless case) 
	\begin{equation*}
		\begin{split}
			\hat{T}_\mrm{e} &= (x_\mrm{q}-x_\mrm{d}')\cos\left(\frac{\pi}{2}-\hat{x}_1+y_4\right)\sin\left(\frac{\pi}{2}-\hat{x}_1+y_4\right)y_3^2 \\  & \ \ \ \ +\hat{x}_3\sin\left(\frac{\pi}{2}-\hat{x}_1+y_4\right)y_3.
		\end{split}
	\end{equation*}
	By applying the trigonometric identities for the sum of angles and the double angle
		as well as the small angle approximation, $\hat{T}_\mrm{e}$ can be expressed as
		\begin{equation*}
			\begin{split}
				\hat{T}_\mrm{e} &= (x_\mrm{q}-x_\mrm{d}')\cos\left(\frac{\pi}{2}-{x}_1+\phi_\mrm{t}\right)\sin\left(\frac{\pi}{2}-{x}_1+\phi_\mrm{t}\right)I_\mrm{t}^2 + \\  & \ \ \ \ +{x}_3\sin\left(\frac{\pi}{2}-{x}_1+\phi_\mrm{t}\right)I_\mrm{t} + \\  & \ \ \ \ + \frac{1}{2}(x_\mrm{q}-x_\mrm{d}')\bigg(\sin\left(\pi-2{x}_1+2\phi_\mrm{t}\right) (2I_\mrm{t}w_3+w_3^2)  + \\  & \ \ \ \ +
				\cos\left(\pi-2{x}_1+2\phi_\mrm{t}\right) 2(w_4-w_{x1})(I_\mrm{t}+w_3)^2 \bigg) +\\
				& \ \ \ \ +  \sin\left(\frac{\pi}{2}-{x}_1+\phi_\mrm{t}\right)(x_3w_3+I_\mrm{t}w_{x3}+ w_{x3}w_3) -\\ & \ \ \ \ -\cos\left(\frac{\pi}{2}-{x}_1+\phi_\mrm{t}\right) (x_3+w_{x3})(I_\mrm{t}+w_3) (w_{x1}-w_4)+ \\
				& \ \ \ \ + \mathcal{O}\left((w_4-w_{x1})^2 \left((I_\mrm{t}+w_3)^2+(I_\mrm{t}+w_3)(x_3+w_{x3})\right)\right)\\
				&= T_\mrm{e} + w_\mrm{Te}  ,
			\end{split}
		\end{equation*}
		where $w_\mrm{Te}$ is obtained by employing the triangle inequality as follows:
		\begin{equation*}
			\begin{split}
				|w_\mrm{Te}| &\leq 
				\frac{1}{2}|x_\mrm{q}-x_\mrm{d}'|\bigg(2I_\mrm{t}^\mrm{max}w_3^\mrm{max}+(w_3^\mrm{max})^2  + \\  &  + 2(w_4^\mrm{max}-w_{x1}^\mrm{max})(I_\mrm{t}^\mrm{max}+w_3^\mrm{max})^2 \bigg) +\\
				&  +  x_3^\mrm{max}w_3^\mrm{max}+I_\mrm{t}^\mrm{max}w_{x3}^\mrm{max}+ w_{x3}^\mrm{max}w_3^\mrm{max} -\\ &   - (x_3^\mrm{max}+w_{x3}^\mrm{max})(I_\mrm{t}^\mrm{max}+w_3^\mrm{max}) (w_{x1}^\mrm{max}-w_4^\mrm{max})+ \\
				&    +  \mathcal{O}\left((w_4-w_{x1})^2 ((I_\mrm{t}+w_3)^2+(I_\mrm{t}+w_3)(x_3+w_{x3}))\right)\\&\leq w_\mrm{Te}^{\mrm{max}}.
			\end{split}
		\end{equation*}
		To analyze the estimation error of the state $x_2$, we define 
		\begin{equation*}
			\begin{split}
				\bar{x}_2 = x_2^I +kx_1- x_2 .
			\end{split}
		\end{equation*}
		By employing \eqref{eq:obsx2}, the dynamics of $\bar{x}_2$ follows as
		\begin{equation}
			\label{eq:x2_err}
			\begin{split}
				\dot{\bar{x}}_2 &= \dot{x}_2^I +k\dot{x}_1- \dot{x}_2 \\
				&=-({\hat \theta_{1}}+k) ({x}^{I}_{2} + k \hat{x}_1) +\hat \theta_{2} (\hat{T}_\mrm{m}-\hat{T}_\mrm{e}) +kx_2- \\ & \ \ \ \ -\dot{x}_2 \pm \hat{\theta}_{1} x_2\\
				&= -({\hat \theta_{1}}+k) ( x_2^I +kx_1 +kw_1- x_2)- \tilde{\theta}_{1} x_2 +\\& \ \ \ \ + \tilde{\theta}_{2}(T_\mrm{m} + w_\mrm{Tm} - T_\mrm{e}- w_\mrm{Te}) +{\theta}_{2}( w_\mrm{Tm} - w_\mrm{Te}) \\
				&= -({\hat \theta_{1}}+k) \bar{x}_2 + \theta_{2}(w_\mrm{Tm} - w_\mrm{Te}) -  ({ \theta_{1}}+k) kw_1 - \\
				& \ \ \ \   - [x_2+kw_1 \ \  T_{e}+w_\mrm{Te}-T_\mrm{m}-w_\mrm{Tm} ] \begin{bmatrix} \		 \tilde{\theta}_{1} \\  \tilde{\theta}_{2}\end{bmatrix}.
			\end{split}
		\end{equation}
		To derive properties of the DREM-based parameter estimator \eqref{eq:para_est} under noisy measurements, we follow the DREM procedure detailed in \cite[Chapter 4.2]{lorenz-meyer_pmu-based_2020-1}) for noiseless measurements. By applying the filter 
		\begin{equation*}
			\mathcal{F} = \frac{\lambda^2}{(\lambda+s)^2},
		\end{equation*}
		with tuning parameters $\lambda>0$ (and where $s$ denotes the Laplace operator) to \eqref{x2} and \eqref{x3}, we define
		\begin{equation}
			\label{eq:lre_def}
			\begin{split}
				z & \coloneqq \mathcal{F}[s^2[\hat{x}_{1}]] , \\
				\bm{\xi}&\coloneqq \begmat{-\mathcal{F}[s[\hat{x}_{1}]]\\ \mathcal{F}[\hat{T}_\mrm{m}-\hat{T}_\mrm{e}]} .
			\end{split}
		\end{equation}
		Due to the linearity of the filters, we can express \eqref{eq:lre_def} as 
		\begin{equation*}
			\begin{split}
				z & = \mathcal{F}[s^2[{x}_{1}]] + \mathcal{F}[s^2[w_{1}]] = z^\mrm{nom} + w_\mrm{z} , \\
				\bm{\xi} &=  \begmat{-\mathcal{F}[s[{x}_{1}]]\\ \mathcal{F}[T_\mrm{m}-{T}_\mrm{e}]}+ \begmat{-\mathcal{F}[s[w_{1}]]\\ \mathcal{F}[w_\mrm{Tm}-w_\mrm{Te}]} = \bm{\xi}^\mrm{nom}+ \bm{w}_{\xi},
			\end{split}
		\end{equation*}
		where the disturbances are bounded via
		\begin{equation*}
			\label{eq:lre_noisy}
			\begin{split}
				|w_\mrm{z}| & \leq w_1^\mrm{max} \|\mathcal{F}[s^2]\|_\infty \leq w_\mrm{z}^\mrm{max}  , \\
				|{w}_{\xi,1}|& \leq  w_1^\mrm{max}\|\mathcal{F}[s]\|_\infty \leq w_\mrm{\xi,1}^\mrm{max},\\
				|{w}_{\xi,2}|& \leq  (w_\mrm{Tm}^\mrm{max} -w_\mrm{Te}^\mrm{max})\|\mathcal{F}\|_\infty \leq w_\mrm{\xi,2}^\mrm{max},	\end{split}
		\end{equation*}
		with the $\mathcal H_\infty$-norm of the proper transfer functions denoted by $\| \cdot \|_\infty$. 
		Thus, we can express the nominal part via the linear regression equation (LRE)
		\begin{equation}
			\label{eq:lre}
			\begin{split}
				z^\mrm{nom}= \bm{\xi}^\mrm{nom} \bm{\theta}.
			\end{split}
		\end{equation}
		The LRE \eqref{eq:lre} is extended applying a linear, single-input 2-output, bounded-input bounded-output (BIBO)-stable operator $\bm{\mathcal{H}}$ yielding 
		\begin{equation} \label{eq:ext_lre} \bm{Z}^\mrm{nom} = 	\bm{\Xi}^\mrm{nom} \bm{\theta}.
		\end{equation}
		In the case of noisy measurements, the nominal parts are approximated by 
		\begin{equation}
			\begin{split}
				\label{eq:regressor_ext}
				\bm{Z} &\coloneqq \bm{\mathcal{H}}[z],\\
				\bm{\Xi} &\coloneqq \bm{\mathcal{H}}[(\bm{\xi})^\top].
			\end{split}
		\end{equation}
		Due to the linearity of the operator $\bm{\mathcal{H}}$, we can express \eqref{eq:regressor_ext} as 
		\begin{equation*}
			\begin{split}
				\bm{Z} & = \bm{\mathcal{H}}[z^\mrm{nom}] + \bm{\mathcal{H}}[w_\mrm{z}] = \bm{Z}^\mrm{nom}+\bm{w}_Z ,\\
				\bm{\Xi} & = \bm{\mathcal{H}}[(\bm{\xi}^\mrm{nom})^\top] + \bm{\mathcal{H}}[(\bm{w_\xi})^\top] \\ & \ \ \ \ = \bm{\Xi}^\mrm{nom} +\bm{w}_\Xi,
			\end{split}
		\end{equation*}
		where the disturbances are bounded by
		\begin{equation*}
			\begin{split}
				|{w}_{Z,i}| & \leq w_\mrm{z}^\mrm{max} \|\bm{\mathcal{H}}\|_\infty  , i = \{1,2\}, \\
				|{w}_{\Xi,1i}|& \leq  w_{\xi,1}^\mrm{max} \|\bm{\mathcal{H}}\|_\infty  , i = \{1,2\}, \\
				|{w}_{\Xi,2i}|& \leq  w_{\xi,2}^\mrm{max}\|\bm{\mathcal{H}}\|_\infty  , i = \{1,2\}, \\	\end{split}
		\end{equation*}
		with $\bm{w}_{\Xi,ji}$ denoting the $j,i$-th entry of the matrix $\bm{w}_{\Xi}$.
		
		In the last step of the DREM procedure, the extended LRE \eqref{eq:ext_lre}, with \eqref{eq:regressor_ext} utilized for the noisy case, is mixed defining
		\begin{equation*} 
			\begin{split}
				\mathcal{Z} &\coloneqq \adj\{\bm{\Xi}\}\bm{Z}, \\
				\Delta \bm{\theta}  &\coloneqq	\adj\{\bm{\Xi}\}\bm{\Xi}  =\det \{\bm{\Xi}\} ,
			\end{split}
		\end{equation*}
		where $\adj\{\cdot\}$ is the adjunct matrix and $\det\{\cdot\}$ is the determinant. In the noiseless case, this yields a decoupled LRE of the form
		\begin{equation} 
			\label{eq:ext_lre_mixed}
			\mathcal{Z}_i^\mrm{nom} =\Delta^\mrm{nom} {\theta}_i.
		\end{equation}
		As $\bm{\Xi} \in \mathbb{R}^{2 \times 2}$, we can decouple the influence of the nominal part and the disturbance in the $\adj\{\bm{\Xi^\mrm{nom}} +\bm{w_\Xi}\}$ and express $\mathcal{Z}$ as 
		\begin{equation*} 
			\begin{split}
				\mathcal{Z} &=(\adj\{\bm{\Xi^\mrm{nom}}\}+ \adj\{\bm{w_\Xi}\}) (\bm{Z}^\mrm{nom}+\bm{w}_Z ) \\ &= \adj\{\bm{\Xi^\mrm{nom}}\}\bm{Z}^\mrm{nom}+\adj\{\bm{\Xi^\mrm{nom}}\}\bm{w}_Z+ \\ & \ \ \ \ + \adj\{\bm{w_\Xi}\} (\bm{Z}^\mrm{nom}+\bm{w}_Z ) 
				\\ & = 	\mathcal{Z}^\mrm{nom} + \bm{w_\mathcal{Z}},
			\end{split} 
		\end{equation*}
		with 
		$$|\bm{w_\mathcal{Z}}| \leq \bm{w_\mathcal{Z}}^\mrm{max},$$
		where the existence of $\bm{w_\mathcal{Z}}^\mrm{max}$ is ensured by the facts that the disturbances $\bm{w}_Z$ and $\bm{w_\Xi}$ are bounded and $\bm{Z}^\mrm{nom}$ as well as $\bm{\Xi^\mrm{nom}}$ remain bounded as they are the outputs of stable filters with bounded inputs. 
		
		The error in the determinant defined as 
		$$ w_\Delta = \Delta- \Delta^\mrm{nom},$$
		can be bounded utilizing \cite[Theorem 2.12]{ipsen_perturbation_2008} via
		\begin{equation*} 
			\begin{split}
				|w_\Delta| \leq 2\|\bm{w}_\Xi\|_2 \ \max \{\|\bm{\Xi^\mrm{nom}}\|_2, \ \|\bm{\Xi}\|_2 \} \leq w_\Delta^\mrm{max}.
			\end{split} 
		\end{equation*}
		Thus, for noisy measurements \eqref{y} the dynamics of the estimation error $\tilde{\theta}_i ,\; i=\{1,2\}$ can be expressed as
		\begin{equation} 
			\label{eq:para_err}
			\begin{split}
				\dot{\tilde{\theta}}_i &= 	\dot{\hat{\theta}}_i- 	\dot{{\theta}}_i = 	\dot{\hat{\theta}}_i \\ &= - \gamma_{i} (\Delta^\mrm{nom}+ w_\Delta)\left((\Delta^\mrm{nom}+ w_\Delta)  \hat \theta_{i} -(\mathcal{Z}_i^\mrm{nom}+ w_{\mathcal{Z},i}) \right) \\
				&= \dot{\tilde{\theta}}_i^\mrm{nom} - \gamma_{i} \left( 2\Delta^\mrm{nom} w_\Delta + w_\Delta^2 \right) \tilde{\theta}_i  - \\ & \ \ \ \ - \gamma_{i} \left(  \left(\Delta^\mrm{nom} w_\Delta + w_\Delta^2\right)\theta_i  - w_{\mathcal{Z},i} \left(  \Delta^\mrm{nom} +w_\Delta  \right)\right). 
			\end{split} 
		\end{equation}
		Combining \eqref{eq:x2_err} and \eqref{eq:para_err}, we can express the error dynamics under noisy measurements in the form of a LTV system 
		\begin{equation}
			\label{eq:err_ad_ob_form}
			\dot{\bm{x}} = A(t)\bm{x}+ \Lambda(t)\bm{x} + \bm{d}(t)
		\end{equation}
		where the detailed description of the time-varying vector $\bm{d}(t)$ as well as the matrices $A(t)$ and $\Lambda(t)$ is provided in \eqref{eq:err_ad_ob_det} and the state vector is defined as $$\bm{x} =	\begin{bmatrix}{\bar{x}}_2 & {\tilde{\theta}}_1 & {\tilde{\theta}}_2\end{bmatrix}^\top  .$$
		\begin{figure*}[!t]
			\begin{equation}
\fontsize{4.5}{6}\selectfont
				\label{eq:err_ad_ob_det}
				\begin{aligned}
					\begin{bmatrix}\dot{\bar{x}}_2 \\ \dot{\tilde{\theta}}_1 \\ \dot{\tilde{\theta}}_2 \end{bmatrix} = 
					\begin{bmatrix}-(k+\hat{\theta}_1) & -x_2 & T_\mrm{m} - T_\mrm{e} \\ 0 & -\gamma_1 \Delta^\mrm{nom}& 0 \\ 0 & 0 & -\gamma_2 \Delta^\mrm{nom}	\end{bmatrix}\begin{bmatrix}{\bar{x}}_2 \\ {\tilde{\theta}}_1 \\ {\tilde{\theta}}_2	\end{bmatrix} 
					+  \begin{bmatrix}0 & -kw_1 & w_\mrm{Tm} - w_\mrm{Te} \\ 0 & -\gamma_1 \left( 2\Delta^\mrm{nom} w_\Delta + w_\Delta^2\right)& 0 \\ 0 & 0 & -\gamma_2 \left( 2\Delta^\mrm{nom} w_\Delta + w_\Delta^2\right)	\end{bmatrix}\begin{bmatrix}{\bar{x}}_2 \\ {\tilde{\theta}}_1 \\ {\tilde{\theta}}_2 	\end{bmatrix} 
					+ \begin{bmatrix}
						\theta_{2}(w_\mrm{Tm} - w_\mrm{Te}) -  ({ \theta_{1}}+k) kw_1 	\\
						- \gamma_{1} \left(  \left(\Delta^\mrm{nom} w_\Delta + w_\Delta^2\right)\theta_1  - w_{\mathcal{Z},1} \left(  \Delta^\mrm{nom} +w_\Delta  \right)\right) \\
						- \gamma_{2} \left(  \left(\Delta^\mrm{nom} w_\Delta + w_\Delta^2\right)\theta_2  - w_{\mathcal{Z},2} \left(  \Delta^\mrm{nom} +w_\Delta  \right)\right) 
					\end{bmatrix}
				\end{aligned}
			\end{equation}
			\hrule 
		\end{figure*}
		In \cite{lorenz-meyer_pmu-based_2020-1}, the nominal system 
		\begin{equation}
			\label{eq:nom_err_sys}
			\dot{\bm{x}} = A(t) \bm{x}
		\end{equation}
		is shown to be globally exponentially stable, provided the regressor $\Delta$ is of persistence of excitation (see \cite[Equation 2.5.3]{sastry_adaptive_1989} for a definition). Consequently, the existence of a Lyapunov function for the nominal system \eqref{eq:nom_err_sys} satisfying \cite[Eq. 9.3 - 9.5]{khalil_nonlinear_2002} is guaranteed by \cite[Theorem 4.14]{khalil_nonlinear_2002}.
		Recalling the assumption of sufficiently small noise bounds, the perturbation of \eqref{eq:err_ad_ob_form} can be bounded as 
		\begin{equation*}
			\|\bm{g}(\bm{x},t)\|_2 \leq \rho \|\bm{x}\|_2 + \zeta,
		\end{equation*}
		with 
		\begin{equation*}
			\begin{split}
				\|\Lambda(t)\|_2 \leq \rho,\\
				\|\bm{d}(t)\|_2 \leq {\zeta}.
			\end{split}
		\end{equation*}
		By invoking \cite[Lemma 9.4]{khalil_nonlinear_2002}, we bound the trajectories of \eqref{eq:err_ad_ob_form} as
		\begin{equation}
			\label{eq:bound_x}
			\begin{split}	\|\bm{x}\|_2 &\leq \sqrt{\frac{c_2}{c_1}}\rho \|\bm{x_0}\|_2 \mrm{e}^{-\alpha(t-t_0)}+ \\ & \ \ \ \ + \frac{c_4\rho\zeta}{2c_1\alpha} \left(1-\mrm{e}^{-\alpha(t-t_0)}\right)\leq w_\theta^\mrm{max},
			\end{split}
		\end{equation}
		where $\bm{x_0}$ is the initial value of $\bm{x}$ at $t_0$, $c_i, i= 1,2,4$ are constants defining bounds on the Lyapunov function as defined in \cite[Eq. 9.3 - 9.5]{khalil_nonlinear_2002} and $\alpha$ is a constant defined in dependence on the bounds of the Lyapunov function \cite[Eq. 9.22]{khalil_nonlinear_2002}.
		
		Finally, the estimation error of the state $x_2$ 
		\begin{equation}
			\label{eq:tilde_x2}
			\begin{split}
				\tilde{x}_2 = \bar{x}_2 +k w_1 ,
			\end{split}
		\end{equation}
		is bounded using the triangular inequality as
		\begin{equation*}
			\begin{split}
				|\tilde{x}_2| \leq |\bar{x}_2| +|k| |w_1| \leq w_\theta^\mrm{max} + k w_1^\mrm{max} \leq w_{x2}^\mrm{max}.
			\end{split}
		\end{equation*} 
		Hence, we have shown that the adaptive observer's error state defined as 
		\begin{equation}
			\label{eq:full_err_state}
			\bm{x}^\mrm{err} \coloneqq	\begin{bmatrix}{\tilde{x}}_2 & {\tilde{\theta}}_1 & {\tilde{\theta}}_2\end{bmatrix}^\top 
		\end{equation}
		is ultimately bounded.
	\end{proof}
	By combining Propositions 1 and 2, we obtain the main result of this paper.
	\begin{corollary}
		Consider the states $x_{i},  i = \{1,2,3\}$ and the parameters ${\theta}_i, i = \{1,2\}$ of the system \eqref{x} as well as the algebraic estimation formula given in \eqref{eq:alg_obs_noise} and the adaptive observer given in \eqref{eq:ad_obs}. Under noisy measurements as in \eqref{y}, considering the estimate of the mechanical torque given in \eqref{eq:Tm_est} and assuming sufficiently small noise bounds, if the regressor $\Delta$ is of persistence of excitation, then the estimation errors defined in \eqref{eq:err:alg_obs} and \eqref{eq:full_err_state} are ultimately bounded. 
	\end{corollary}
	
	\section{Simulation results}
	In this section, we illustrate the impact of PMU measurements with Gaussian and non-Gaussian noise on the method's performance. For this, we show simulations using the well-known New England IEEE-39 bus system with the parameters provided in \cite{hiskens_ieee_nodate}. The SGs are represented by the third-order flux-decay model and are equipped with AVRs and PSSs according to Figure \ref{fig:SG_model}. Consequently, each SG is modeled by a 9-dimensional model, details of which are given in \cite[Eq. (18)]{lorenz-meyer_pmu-based_2020-1}. 
	\begin{figure}
		\centering
		\includegraphics[width=1\linewidth]{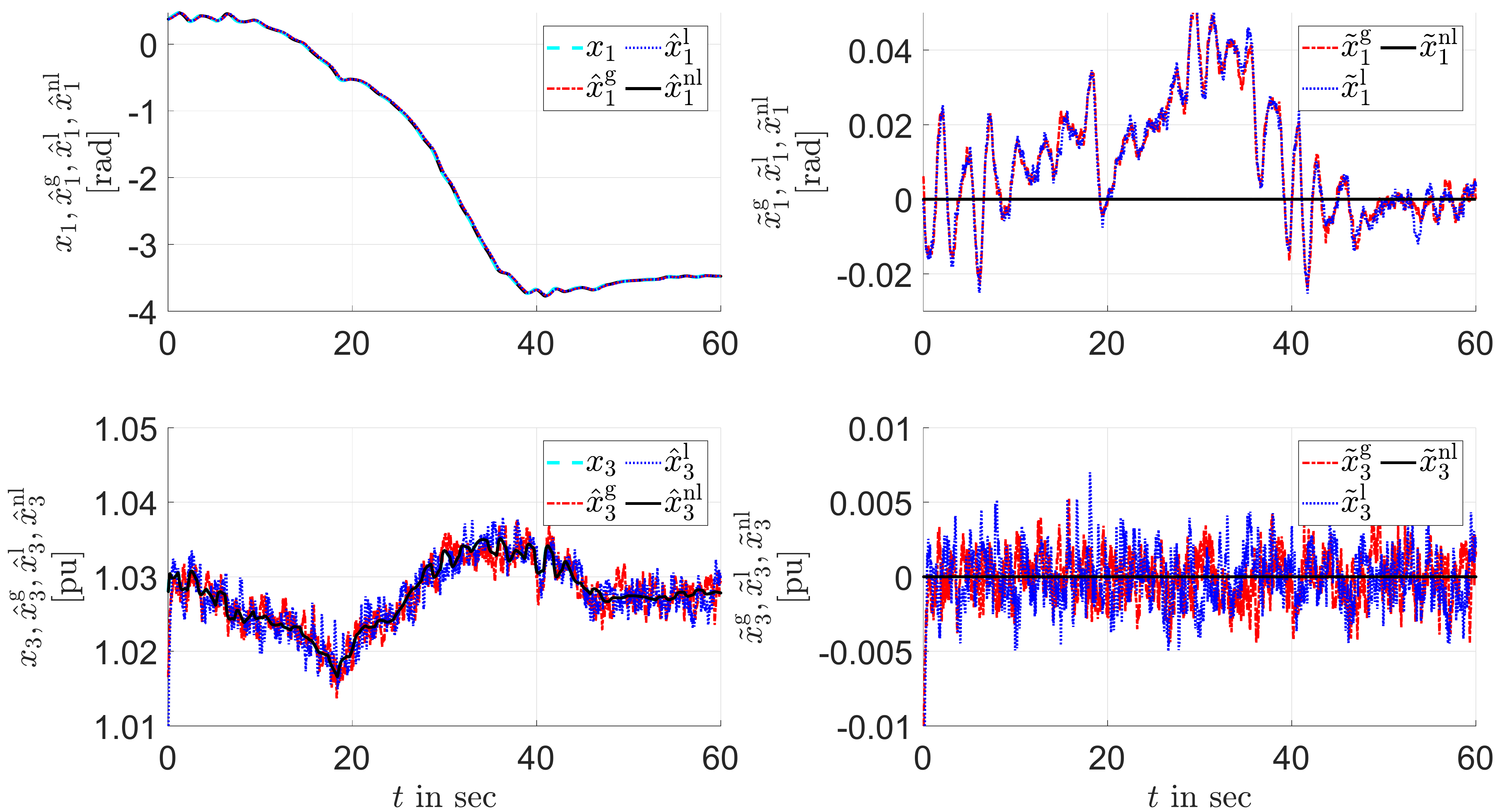}
		\caption{Algebraic state estimation for $x_{1}$ and $x_{3}$ of generator 10 in the presence of load variations. Owing to the algebraic structure of this estimation, no initial conditions are required.}
		\label{fig:adp_obs}
	\end{figure}
	
	We assume a PMU installed at the terminal bus of generator 10 and consider three cases to illustrate the performance under realistic conditions:
	\begin{enumerate}
		\item A noiseless case, denoted with superscript $\{\cdot\}^\mrm{nl}$.
		\item A case, where the PMU measurements are disturbed by zero mean Gaussian noise, denoted with superscript $\{\cdot\}^\mrm{g}$. 
		\item A case, where the PMU measurements are disturbed by zero mean Laplacian noise, denoted with superscript $\{\cdot\}^\mrm{l}$.  Usage of Laplacian distributions to mimic realistic PMU measurement noise is recommended in \cite{wang_assessing_2018}.
	\end{enumerate}
	For the cases with Gaussian as well as Laplacian measurement noise, the signal to noise ratio is set to 45 dB according to \cite{brown_characterizing_2016}.
	
	The simulation scenario consists of minor load variations in the system yielding frequency variations within $60\pm0.05$~[Hz] and is thus consistent with regular operation of transmission grids \cite{weissbach_verbesserung_2009}. In Figure \ref{fig:adp_obs}, the results of the algebraic state estimation results according to \eqref{eq:alg_obs_noise} are shown for all three cases. Owing to the algebraic structure, initial conditions are not required.  As can be seen, the method is able to accurately track the states $x_1$ and $x_3$ of generator 10 with a bounded error. 
	In Figure \ref{fig:DREM}, the results of the DREM-based parameter estimator and the adaptive observer according to \eqref{eq:ad_obs} are shown for all three cases. It can be seen that the parameter estimator converges over a longer period in the cases of noisy PMU measurements. Albeit, after $t\approx 42$ sec the parameters estimates for $a_1$ and $a_2$ converge with a bounded error. Consequently, the state estimation error also converges with a bounded error.  
	\begin{figure}
		\centering
		\includegraphics[width=1\linewidth]{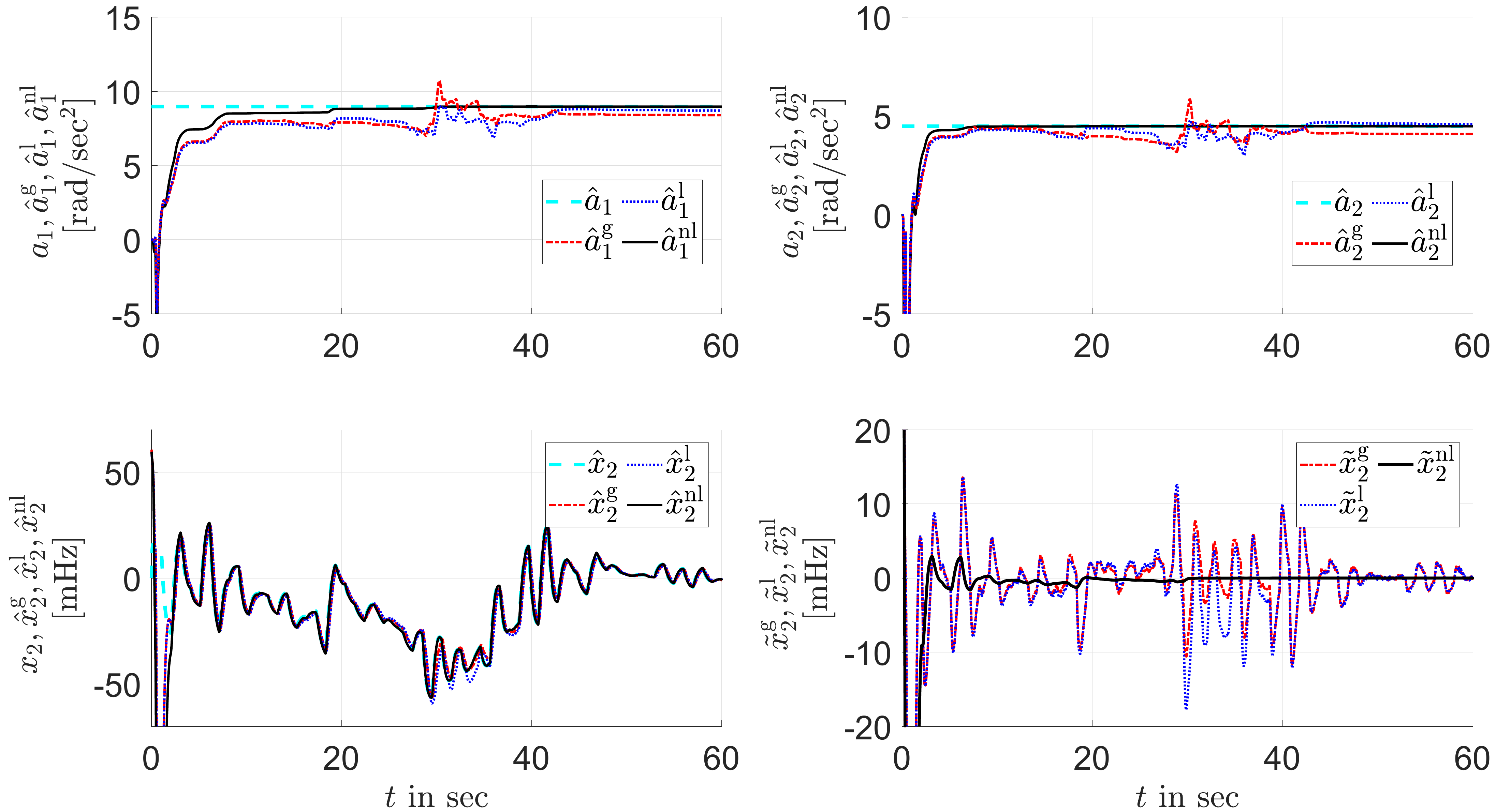}
		\caption{Simulation results of the DREM-based parameter estimation and the adaptive observer for $x_{2}$ of generator 10 in the presence of load variations. }
		\label{fig:DREM}
	\end{figure}
	
	\section{Conclusion}
	Building on the previous work \cite{lorenz-meyer_pmu-based_2020-1,lorenz-meyer_dynamic_2022}, in this note the robustness of the algebraic estimator \eqref{eq:alg_obs_noise} as well as the adaptive observer \eqref{eq:ad_obs} against measurement noise is investigated. For this, an analytical analysis of the impact of measurement noise on the estimates is presented. It is found that a bounded noise results in a bounded estimation error. For nonlinear observers this important robustness property cannot be guaranteed in general \cite{shim_nonlinear_2016}, which underpins the relevance of asserting it. 
	Finally, the impact of PMU measurements with Gaussian and non-Gaussian noise on the method's performance is shown in realistic simulations using the New England IEEE-39 bus system. Thereby, the bounded state and parameter estimation error of the method in the presence of noisy measurements is illustrated.
	\section*{Acknowledgements} The authors would like to thank Juan G. Rueda-Escobedo for many helpful comments on the topics of this paper.

		\bibliographystyle{IEEEtran}
		\bibliography{Diss_bib}

	\end{document}